\documentclass[11pt]{article}

\usepackage[utf8]{inputenc}
\usepackage[T1]{fontenc}
\usepackage{amsmath, amsfonts, amssymb}
\usepackage{amsthm}

\renewenvironment{proof}{\par\noindent\textit{Proof.} }{\hfill$\square$\par}
\usepackage{graphicx}
\usepackage{tikz}
\usetikzlibrary{arrows.meta, positioning}

\usepackage{graphicx}

\usepackage{authblk}

\usepackage[utf8]{inputenc}
\usepackage{multirow} 
\usepackage{array} 

\title{\bf On Brain as a Mathematical Manifold:\
Neural Manifolds, Sheaf Semantics, and Leibnizian Harmony}

\author{\Large Takao Inou\'{e}}

\affil{\large Faculty of Informatics, Yamato University, \\ Osaka, Japan\footnote{Email: inoue.takao@yamato-u.ac.jp; \\ Personal Email: takaoapple@gmail.com \\ [I prefer my personal email address for correspondence.]}} 

\date{January 16, 2026. Preprint Version 1.1\footnote{I am still in the state of trial and error for this paper. This is the very first version. However this paper gives a consistent perspective for the brain pathology and certain interpretation for the thought of Leibniz's monadology \cite{Leibniz1714} from a point of view of sheaf theory, although it is right or not. This manuscript is posted as a preprint to establish priority.}}

\newtheorem{theorem}{Theorem}[section]

\newtheorem{definition}{Definition}[section]

\newtheorem{remark}{Remark}[section]

\begin{document}
\maketitle

\begin{abstract}
We present a mathematical and philosophical framework in which brain function is modeled using sheaf theory over neural state spaces. Local neural or cognitive functions are represented as sections of a sheaf, while global coherence corresponds to the existence of global sections. Brain pathologies are interpreted as obstructions to such global integration and are classified using tools from sheaf cohomology. The framework builds on the neural manifold program in contemporary neuroscience and on standard results in sheaf theory, and is further interpreted through a Leibnizian lens \cite{Churchland2012, Leibniz1714,  MacLaneMoerdijk, Perich2025}. This paper is intended as a conceptual and formal proposal rather than a complete empirical theory.\\
\smallskip

\end{abstract}

\bigskip

\tableofcontents

\section{Introduction}
Recent advances in systems neuroscience suggest that population-level neural activity often concentrates on low-dimensional manifolds embedded in high-dimensional neural state spaces. This \emph{neural manifold} perspective has proven fruitful in motor control, perception, and cognition. However, from a philosophical standpoint, the precise conceptual status of such manifolds remains underdetermined: Are they merely empirical summaries, or do they possess a deeper structural and semantic role in theories of mind and brain?

The present paper argues for the latter. We propose a mathematically explicit and philosophically interpretable framework in which brain function is understood through sheaf theory on neural manifolds. Our central claim is organized by a clear \emph{A $\rightarrow$ B} structure. (A) Locally defined neural or cognitive functions are modeled as sections of a sheaf over open subsets of a brain-related space. (B) Globally coherent brain function—such as unified perception, action, or conscious experience—is identified with the existence of a global section obtained via sheaf-theoretic gluing.

This formulation allows us to make precise sense of integration and its failures. In particular, when local functional coherence does not extend globally, the obstruction is not merely empirical but structural, and is naturally classified by sheaf cohomology. On this view, brain pathologies are interpreted as mathematically identifiable failures of local-to-global integration.

Philosophically, this framework offers a novel bridge between contemporary neuroscience and classical rationalist metaphysics. Drawing on Leibniz's notion of monads and pre-established harmony, we interpret local sections as partial perspectives and the sheaf gluing axiom as a formal counterpart of harmony without direct causal interaction. This is not intended as historical exegesis, but as a structural analogy that clarifies the conceptual stakes of global coherence in cognitive systems.

The contribution of this paper is therefore threefold: (i) it provides a rigorous semantic interpretation of neural manifolds using sheaf theory; (ii) it offers a cohomological classification of brain pathology as structural obstruction; and (iii) it situates these results within a broader philosophical context suitable for interdisciplinary dialogue in philosophy of science and philosophy of mind.

This paper is intended as a conceptual and formal proposal rather than a complete empirical theory.

\section{Related Work}

The neural manifold hypothesis has been developed extensively in systems neuroscience (cf. \cite{chung2021, Perich2025b, Perich2025c, Safaie2023}). Seminal work by Churchland et al. \cite{Churchland2012} demonstrated low-dimensional population dynamics in motor cortex, while Gallego et al. \cite{Gallego2017} emphasized manifolds as task-relevant control spaces. Perich et al. \cite{Perich2025} further generalized this viewpoint, arguing that manifold structure underlies flexible cognition beyond motor control.

Topological and geometric approaches to neural data analysis—including persistent homology and information geometry—have also gained prominence. However, these approaches typically lack an explicit local-to-global semantic framework. Sheaf-theoretic methods, previously applied to sensor fusion and distributed systems, provide precisely such a framework. The present work extends these ideas by introducing a sheaf semantics of brain function and pathology, enriched by a Leibnizian philosophical interpretation.

\section{Brain-Related Spaces and Neural Manifolds}

\begin{definition}[Brain-Related Space]
A \emph{brain-related space} $\mathcal{B}$ is a topological space whose points represent neural configurations or cognitive states, and whose open sets correspond to functionally coherent neural subsystems.
\end{definition}

\begin{definition}[Neural Sheaf]
Let $\mathcal{B}$ be a brain-related space. A \emph{neural sheaf} $\mathcal{F}$ assigns to each open set $U \subseteq \mathcal{B}$ a set $\mathcal{F}(U)$ of neural or cognitive functions observable on $U$, together with restriction maps satisfying the sheaf axioms.
\end{definition}

\section{Sheaf Semantics of Brain Function}

\begin{theorem}[Global Brain Function]
\footnote{For the local-to-global principle formalized here, see standard treatments of sheaf theory such as Mac~Lane and Moerdijk (1992, Chapter~III) \cite{MacLaneMoerdijk} or Kashiwara and Schapira (2006, Section~2) \cite{KashiwaraSchapira}.}
A globally coherent brain function corresponds to the existence of a global section
$$
s \in \mathcal{F}(\mathcal{B}).
$$
\end{theorem}

\begin{proof}
We provide a detailed argument based on the axioms of sheaf theory.

Let $\{U_i\}_{i\in I}$ be an arbitrary open cover of $\mathcal{B}$. For each $i$, assume that $s_i \in \mathcal{F}(U_i)$ represents a locally coherent neural or cognitive function. Global coherence of brain function means that these local functions are mutually consistent wherever their domains overlap. Formally, this requires that
$$
s_i|_{U_i \cap U_j} = s_j|_{U_i \cap U_j} \quad \text{for all } i,j \in I.
$$

By the gluing axiom of sheaf theory, the existence of such a compatible family of local sections implies the existence of a unique global section $s \in \mathcal{F}(\mathcal{B})$ satisfying $s|_{U_i} = s_i$ for all $i$. This global section represents a single integrated functional state extending across all neural subsystems.

Conversely, if a global section $s \in \mathcal{F}(\mathcal{B})$ exists, then its restrictions $s|_{U_i}$ automatically define a compatible family of local sections. Hence, the existence of a global section is both necessary and sufficient for global coherence.
\end{proof}

\begin{remark}[Standard Sheaf-Theoretic Fact]
The local-to-global correspondence used in the above proof is a standard consequence of the sheaf gluing axiom. See, for example, Mac Lane and Moerdijk, \emph{Sheaves in Geometry and Logic}, Chapter II \cite{MacLaneMoerdijk}, or Kashiwara and Schapira, \emph{Categories and Sheaves}, Section 2.3 \cite{KashiwaraSchapira}. The present contribution lies not in this mathematical fact itself, but in its interpretation as a model of global brain function.
\end{remark}

\section{Pathology as Cohomological Obstruction}
The passage from locally defined sections on an open cover to a global section is mediated by compatibility conditions on overlaps. When these conditions fail, the resulting \v{C}ech cocycle provides an explicit representation of pathological obstruction.

\begin{theorem}[Pathological Obstruction]
\footnote{This argument follows the standard \v{C}ech cohomological description of obstructions to gluing; see Bott and Tu (1982, Chapter~9) \cite{BottTu} .}
Let $\mathcal{F}$ be a neural sheaf over a brain-related space $\mathcal{B}$. If locally defined sections fail to glue to a global section, then the obstruction is represented by a nontrivial cohomology class in $H^1(\mathcal{B},\mathcal{F})$.
\end{theorem}

\begin{proof}
Let $\{U_i\}$ be an open cover of $\mathcal{B}$ and $s_i \in \mathcal{F}(U_i)$ local sections. Define a \v{C}ech 1-cochain
$$
c_{ij} = s_i|_{U_i\cap U_j} - s_j|_{U_i\cap U_j}.
$$

If the family ${s_i}$ were globally compatible, then $c_{ij}=0$ for all $i,j$. In general, the collection ${c_{ij}}$ satisfies the cocycle condition on triple overlaps, hence defines a cohomology class $[c] \in H^1(\mathcal{B},\mathcal{F})$.

If this class were trivial, then $c$ would be a coboundary and the local sections could be modified to yield a compatible family admitting a global section. Therefore, a nontrivial class $[c]$ precisely measures the failure of global integration. Interpreted neurocognitively, this corresponds to pathological dysfunction.
\end{proof}
$$ $$

\begin{theorem}[Cohomological Classification]
\footnote{The interpretation of $H^0$ as global sections and higher cohomology as obstructions is classical; see Mac~Lane and Moerdijk (1992) \cite{MacLaneMoerdijk} and Weibel (1994, Chapter~5) \cite{Weibel}.}
For a neural sheaf $\mathcal{F}$ on $\mathcal{B}$:
\begin{itemize}
\item $H^0(\mathcal{B},\mathcal{F})$ corresponds to normal global integration.
\item Nontrivial $H^k(\mathcal{B},\mathcal{F})$ for $k>0$ correspond to distributed or higher-order pathologies.
\end{itemize}
\end{theorem}

\begin{proof}
By definition, $H^0(\mathcal{B},\mathcal{F})$ is the space of global sections. Higher cohomology groups measure successive failures of local-to-global compatibility across multiple overlaps. These failures cannot be resolved by local corrections alone and thus represent structural dysfunctions at increasing levels of integration.
\end{proof}

\subsection{Correspondence Table}


\begin{center}
\renewcommand{\arraystretch}{1.3}
\begin{tabular}{ll}

\textbf{Sheaf-Theoretic Concept} & \textbf{Philosophical / Neuroscientific Interpretation} \\ 
Open set $U \subseteq \mathcal{B}$ & Local neural or cognitive subsystem \\ 
Local section $s \in \mathcal{F}(U)$ & Local functional / semantic state \\ 
Restriction map & Inter-module consistency constraint (Morphism) \\ 
Global section & Integrated brain function / Unified consciousness \\ 
Sheaf gluing axiom & Leibnizian pre-established harmony \\ 
$H^0(\mathcal{B},\mathcal{F})$ & Normal global integration \\ 
$H^{k>0}(\mathcal{B},\mathcal{F})$ & Distributed or higher-order pathology (Obstruction) \\ 
\end{tabular}
\end{center}

This correspondence makes explicit how classical metaphysical notions and contemporary neuroscience are unified through sheaf-theoretic structure.

\section{From Open Covers to Cohomological Obstruction}

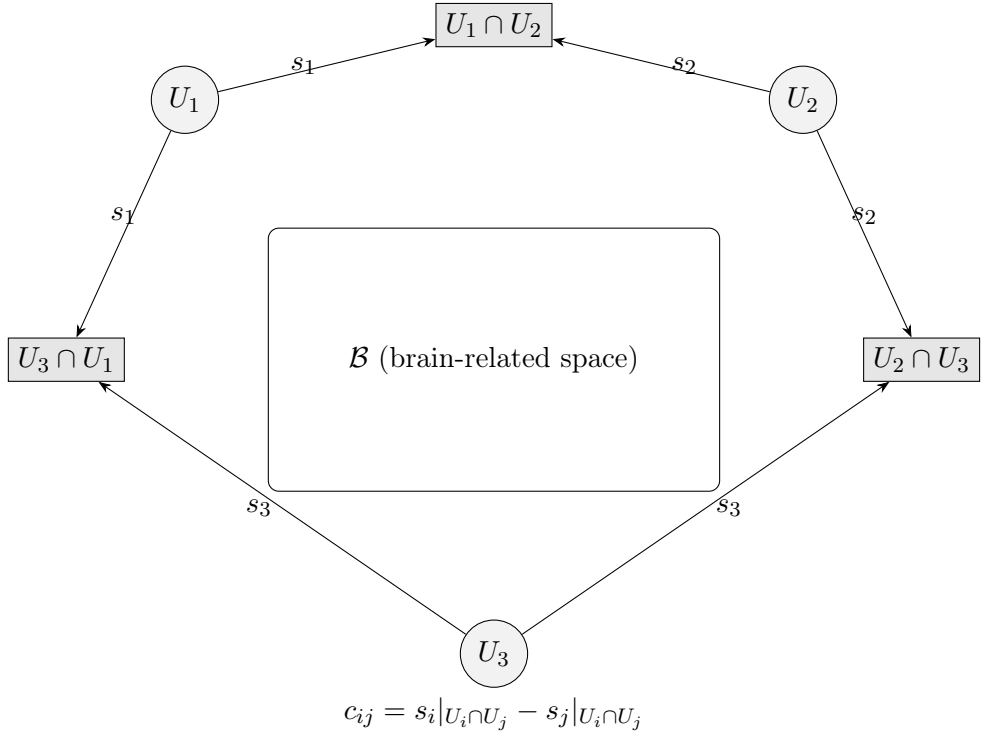
\begin{figure}[h]
\centering
\begin{tikzpicture}[node distance=2.8cm, >=Stealth]

\node (B) [draw, rounded corners, minimum width=6cm, minimum height=3.5cm] {$\mathcal{B}$ (brain-related space)};

\node (U1) [draw, circle, fill=gray!10, above left=of B, xshift=1.2cm, yshift=-0.6cm] {$U_1$};
\node (U2) [draw, circle, fill=gray!10, above right=of B, xshift=-1.2cm, yshift=-0.6cm] {$U_2$};
\node (U3) [draw, circle, fill=gray!10, below=of B, yshift=1.1cm] {$U_3$};

\node (O12) [draw, rectangle, fill=gray!20, above=of B, yshift=-0.4cm] {$U_1 \cap U_2$};
\node (O23) [draw, rectangle, fill=gray!20, right=of B, xshift=-0.9cm] {$U_2 \cap U_3$};
\node (O31) [draw, rectangle, fill=gray!20, left=of B, xshift=0.9cm] {$U_3 \cap U_1$};

\draw[->] (U1) -- node[left] {$s_1$} (O12);
\draw[->] (U2) -- node[right] {$s_2$} (O12);

\draw[->] (U2) -- node[above] {$s_2$} (O23);
\draw[->] (U3) -- node[right] {$s_3$} (O23);

\draw[->] (U3) -- node[left] {$s_3$} (O31);
\draw[->] (U1) -- node[above] {$s_1$} (O31);

\node (C) [below=of B, yshift=0.2cm] {$c_{ij} = s_i|_{U_i\cap U_j} - s_j|_{U_i\cap U_j}$};

\end{tikzpicture}
\caption{From an open cover $\{U_i\}$ of a brain-related space $\mathcal{B}$ to a \v{C}ech 1-cocycle. Local sections $s_i$ agree on overlaps in the healthy case. Pathology corresponds to a nontrivial cocycle obstructing global gluing.}\label{fig:open-cover-cocycle}
\end{figure}


\begin{figure}[htbp]
\centering
\begin{tikzpicture}[scale=1.05]

\definecolor{openset}{RGB}{100,160,220}
\definecolor{overlap}{RGB}{140,200,160}
\definecolor{cocycle}{RGB}{200,80,80}

\newcommand{\brainmanifold}[4]{
  \draw[thick]
    plot [smooth cycle, tension=0.8]
    coordinates {
      (#1-2,0) (#1-1,1.2) (#1+0.8,1.4)
      (#1+2,0.6) (#1+1.5,-1.2) (#1-0.5,-1.5)
    };

  \fill[openset, opacity=0.35] (#1-0.8,0.4) circle (0.9);
  \fill[openset, opacity=0.35] (#1+0.9,0.3) circle (0.9);
  \fill[openset, opacity=0.35] (#1+0.1,-0.6) circle (0.9);

  \fill[overlap, opacity=0.45] (#1+0.05,0.25) circle (0.45);
  \fill[overlap, opacity=0.45] (#1+0.45,-0.25) circle (0.45);

  #2

  \node at (#1, -2.3) {\small (#3) #4};

  \node[align=left, text width=4.8cm] at (#1, -3.3) {%
    \scriptsize
    \textcolor{openset}{Open sets $U_i$}: local domains\\
    \textcolor{overlap}{Overlaps $U_i\cap U_j$}: integration\\
    \textcolor{cocycle}{Arrows}: transitions / cocycles
  };
}

\begin{scope}[shift={(0,7)}]
\brainmanifold{0}{
  \draw[->, thick, cocycle] (-0.4,0.5) -- (0.4,0.5);
  \draw[->, thick, cocycle] (0.7,0.1) -- (0.3,-0.4);
  \draw[->, thick, cocycle] (-0.1,-0.4) -- (-0.4,0.2);
}{a}{Normal integration}
\end{scope}

\begin{scope}[shift={(0,0)}]
\brainmanifold{0}{
  \draw[->, thick, cocycle] (-0.4,0.5) -- (0.4,0.5);
  \draw[->, thick, cocycle, dashed] (0.7,0.1) -- (0.3,-0.4);
  \draw[->, thick, cocycle] (-0.1,-0.4) -- (-0.4,0.2);
}{b}{Local degeneration}
\end{scope}

\begin{scope}[shift={(0,-7)}]
\brainmanifold{0}{
  \draw[->, thick, cocycle] (-0.4,0.5) -- (0.4,0.5);
  \draw[->, thick, cocycle] (0.3,-0.4) -- (0.7,0.1);
  \draw[->, thick, cocycle, dashed] (-0.1,0.2) -- (-0.5,0.6);
}{c}{Global inconsistency}
\end{scope}

\begin{scope}[shift={(4.6,7.2)}]
  \draw[openset, fill=openset, opacity=0.35] (0,0) rectangle (0.4,0.3);
  \node[right] at (0.45,0.15) {\scriptsize open set $U_i$};

  \draw[overlap, fill=overlap, opacity=0.45] (0,-0.45) rectangle (0.4,-0.15);
  \node[right] at (0.45,-0.3) {\scriptsize overlap};

  \draw[->, cocycle, thick] (0,-0.8) -- (0.4,-0.8);
  \node[right] at (0.45,-0.8) {\scriptsize transition};

  \draw[->, cocycle, thick, dashed] (0,-1.15) -- (0.4,-1.15);
  \node[right] at (0.45,-1.15) {\scriptsize degeneration};
\end{scope}

\end{tikzpicture}
\caption{
Brain modeled as a manifold with an open cover.
(a) Consistent cocycle conditions yield coherent global integration.
(b) Local degeneration of transitions, corresponding to aphasia or agnosia.
(c) Global inconsistency of cocycles, corresponding to schizophrenia-like disorders.
}
\end{figure}



From $\{U_i\}$ of a brain-related space $\mathcal{B}$ to a \v{C}ech 1-cocycle. Local sections $s_i$ agree on overlaps in the healthy case. Pathology corresponds to a nontrivial cocycle obstructing global gluing.

$$ $$

\section{Clinical Case Studies and Literature Alignment}

We briefly align the abstract sheaf-theoretic classification with representative
clinical and neuroscientific literature, emphasizing that no new empirical claims
are made.

\paragraph{Aphasia.}
Classical disconnection accounts of aphasia (including Broca-type and conduction
aphasia) describe a loss or severe restriction of locally expressible linguistic
functions. In the present framework, this corresponds to a degenerate case in which
certain open sets admit no nontrivial local sections. This interpretation is
consistent with the classical neuropsychological synthesis of Geschwind
\cite{Geschwind1965}.

\paragraph{Agnosia.}
Visual and associative agnosias are typically analyzed as failures of perceptual
recognition despite preserved early sensory processing. Such cases correspond to
locally defined sections that are trivial or semantically impoverished, rather than
globally incompatible. This aligns with standard cognitive-neuropsychological
accounts such as Farah \cite{Farah2004}.

\paragraph{Schizophrenia.}
In contrast, schizophrenia is widely modeled as a disorder of integration rather
than local loss, most prominently in the dysconnection hypothesis. Here, local
cognitive contents persist but fail to cohere globally. This is naturally modeled
by a nontrivial cohomology class obstructing global sections, in line with
integrative accounts such as Friston \cite{Friston1998}.

\section{Leibnizian Interpretation}

\subsection{Methodological Remark}
Before introducing the correspondence table below in 8.2, it is important to clarify the methodological status of the present framework. The sheaf-theoretic modeling proposed in this paper is neither a reduction of neurobiological phenomena to abstract mathematics, nor a merely metaphorical analogy between brain function and mathematical structure. Rather, it should be understood as a form of semantic modeling in the sense of philosophy of science.

In this approach, mathematical objects—such as topological spaces, sheaves, and cohomology groups—do not purport to describe the material substrate of the brain directly. Instead, they provide a precise language for articulating relations of dependence, compatibility, and integration among locally defined functional states. The distinction between local and global sections is thus not ontological but structural: it captures how partial cognitive or neural perspectives may or may not integrate into a unified functional whole.

Accordingly, the Leibnizian vocabulary introduced here should not be read as a historical or doctrinal claim about Leibniz's metaphysics. The notions of monads and pre-established harmony function instead as conceptual lenses that illuminate the structural role played by local sections and gluing conditions. The philosophical value of this analogy lies in its ability to clarify what kind of unity is at stake in global brain function—namely, unity without direct causal fusion, but achieved through systematic compatibility constraints.

Seen in this light, the correspondence table below is not a claim of identity between neuroscience and metaphysics, but a controlled translation between mathematical structure and interpretative vocabulary. This translation is intended to facilitate interdisciplinary dialogue while preserving formal rigor.

Local sections are interpreted as monadic perspectives, while the sheaf gluing axiom formalizes pre-established harmony. Pathology corresponds to the mathematical failure of harmony, captured cohomologically.

\subsection{Leibnizian Interpretation: Monads, Atlases, and Brain}

In Leibniz's Monadology \cite{Leibniz1714}, reality consists of monads—self-contained entities providing partial perspectives on the whole. We interpret neural subsystems as monad-like units, each encoding a local representation of the brain's state.

From a geometric perspective, these monads correspond to local charts on the neural manifold. The brain itself forms an atlas integrating these charts via transition functions, while sheaf restriction maps ensure consistency across overlaps.

Brain pathology arises when these transition functions fail or when local sections cannot be coherently glued, resulting in cohomological obstructions. Thus, Leibnizian monads, manifold atlases, and sheaf semantics converge into a unified framework for brain interpretation.

\begin{theorem} (Monadological Interpretation).
Let each local chart of the neural manifold be interpreted as a monad in the Leibnizian
sense. Then the sheaf-theoretic gluing condition formalizes the principle of
pre-established harmony as global coherence of sections.
\end{theorem}

Proof [Proof Sketch]
\medskip

In Leibnizian monadology, each monad provides a local, self-contained perspective on the
world, while global order arises not through direct interaction but through harmony among
these perspectives. Analogously, in the sheaf-theoretic framework, each local chart or open
set supports local sections representing functional states.

The sheaf gluing condition ensures that when these local sections are mutually compatible,
they give rise to a unique global section. This mirrors the notion of pre-established
harmony, where global coherence emerges from local consistency without direct causal
exchange. Failures of gluing correspond to breakdowns of harmony, which, in the present
framework, are identified with pathological conditions. $\Box$

\subsection{Connection to Leibnizian Monadology}

This theorem naturally incorporates the following correspondence with the principles of Monadology:

\begin{center}
\renewcommand{\arraystretch}{1.3}
\begin{tabular}{ll}

\textbf{Sheaf Theory} & \textbf{Leibnizian Monadology} \\ 
Local Section & Local representation within a Monad \\ 
Compatibility & Pre-established Harmony (\textit{Harmonie préétablie}) \\ 
Global Section & Coherent representation of the World \\ 
\end{tabular}
\end{center}

\noindent
In this framework, the mathematical requirement for local data to be compatible mirrors the Leibnizian concept of a "Pre-established Harmony," where individual monads reflect the same universe from their own perspectives.

\subsection{Discussion: The Brain as a Machine of Harmony}

The proposed framework invites a fundamental shift in our understanding of neural architecture, transitioning from a view of the brain as a mere signal processor to a \textbf{Leibnizian Machine of Harmony}.

\subsubsection{The Mathematical Restoration of the Monad}
In classical Leibnizian metaphysics, the Monad is "windowless," reflecting the universe from an internal logic. In our sheaf-theoretic model, this is mirrored by the \textbf{Local Section} over an open set $U_i$. While neural modules (visual, auditory, or motor) may appear functionally isolated, they are bound by the \textbf{Sheaf Condition}. This condition acts as the modern successor to \textit{Pre-established Harmony}: it is the mathematical necessity that local "realities" must coincide on their intersections to permit a global existence.

\subsubsection{Consciousness as Global Coherence}
We posit that \textbf{Consciousness} is not a localized "hub" but the realization of a \textbf{Global Section}. When the brain satisfies the gluing axioms across its entire manifold $\mathcal{B}$, a unified phenomenal field emerges. This provides a rigorous answer to the binding problem: the "unity of experience" is the topological success of $H^0(\mathcal{B}, \mathcal{F})$.

\subsubsection{Topological Diagnosis of the Self}
By defining pathology as a \textbf{Cohomological Obstruction}, we move beyond descriptive psychiatry into topological ontology. 
\begin{itemize}
    \item A "healthy" mind is one where $H^{k>0} = 0$, meaning every local thought is part of a possible global whole.
    \item A "pathological" mind, such as in schizophrenia, possesses a non-trivial $H^1$. This implies a "loop" of local consistency that is globally paradoxical—a structural discord where the harmony of the monads is broken.
\end{itemize}

\subsubsection{Concluding Remark}
Ultimately, the brain's primary function is the \textbf{computation of a Global Section}. It is a biological apparatus designed to solve the problem of harmony. When we think, perceive, or exist, we are essentially "gluing" the fragmented, monadic perspectives of our neural subsystems into a single, consistent world.

\section{Summary}
This paper proposes a formal and conceptual framework for describing global integration and its breakdown based on sheaf-theoretic structures.
The brain is idealized as a manifold equipped with an open cover, where local functional domains are represented as sections, and their coordination is governed by gluing and cocycle conditions.
From this perspective, coherent cognitive integration corresponds to the existence of compatible global sections, while local degeneration or global inconsistency of the sheaf structure provides a structural model for brain pathologies such as aphasia, agnosia, and schizophrenia.
The framework is developed independently of historical doctrines; however, it can be applied to reinterpret Leibniz’s notion of pre-established harmony as a condition of global coherence arising from locally consistent structure.
Rather than offering an empirical theory of neural processes, the paper aims to articulate a sheaf-theoretic mode of structural explanation for cognitive unity and its pathological disruption.

\section{Conclusion}
This paper began with a conceptual question motivated by recent developments in systems neuroscience: what is the theoretical status of neural manifolds, and how should the relation between local neural function and global cognitive integration be understood? While empirical studies have successfully identified low-dimensional structure in neural population activity, the conceptual and semantic implications of such structure have remained insufficiently articulated.

By introducing a sheaf-theoretic framework over brain-related spaces, we have proposed a precise answer to this question. Locally defined neural or cognitive functions are modeled as sections over open subsets, while globally coherent brain function corresponds to the existence of a global section obtained via gluing. This formulation makes explicit the structural conditions under which integration is possible, and—crucially—the ways in which it can fail.

From this perspective, brain pathology is not treated merely as localized dysfunction or noise, but as a mathematically identifiable obstruction to global integration. Sheaf cohomology provides a principled means of classifying such obstructions, thereby distinguishing normal integration from distributed or higher-order pathologies. The resulting picture reframes classical clinical phenomena in terms of local-to-global incompatibility rather than isolated deficits.

Philosophically, returning to the initial question, neural manifolds emerge not simply as geometric artifacts of data analysis but as carriers of semantic structure. The Leibnizian interpretation offered here—understood methodologically rather than historically—clarifies how unity can arise from multiplicity without reduction, through compatibility constraints rather than causal fusion. In this sense, the sheaf-theoretic notion of gluing functions as a formal analogue of pre-established harmony.

We conclude by emphasizing that the present framework is not a final theory of brain function, but a conceptual scaffold. Its value lies in providing a common language in which empirical neuroscience, mathematics, and philosophy can meaningfully interact. Future work may refine the choice of brain-related spaces, enrich the categorical structure of neural sheaves, or connect cohomological invariants more directly with empirical measures of dysfunction. What remains constant is the guiding question to which we now return: how local neural perspectives are able—or fail—to constitute a globally coherent mind.

The philosophical orientation of this paper is continuous with the author’s related work \cite{Inoue2026, Inoue2026a, Inoue2026b}, in which similar structural concerns are explored from complementary perspectives. Together, through the notion of name and meaning, these works are related to articulating a unified conceptual framework for understanding coherence, integration, and their failure across logical, mathematical, and cognitive domains.

\appendix

\section{Appendix: A Čech Cohomology Calculation for Neural Sheaves}

\noindent\emph{No originality is claimed in this calculation.}\

This appendix provides a minimal and self-contained illustration of how a Čech cohomology class arises as an obstruction to global integration, complementing the proofs in the main text.

Let $\mathcal{B}$ be a brain-related space covered by two open sets $U$ and $V$ such that $\mathcal{B}=U\cup V$ and $U\cap V\neq\emptyset$. Let $\mathcal{F}$ be a neural sheaf on $\mathcal{B}$.

Assume that we are given local sections
$$
s_U\in\mathcal{F}(U), \qquad s_V\in\mathcal{F}(V).
$$
On the overlap $U\cap V$, define the Čech 1-cochain
$$
c_{UV} = s_U|_{U\cap V} - s_V|_{U\cap V} \in\mathcal{F}(U\cap V).
$$

If $c_{UV}=0$, then the two local sections agree on the overlap and, by the sheaf gluing axiom, define a global section $s\in\mathcal{F}(\mathcal{B})$. In this case, the corresponding cohomology class vanishes.

If $c_{UV}\neq 0$, then no global section exists that restricts to both $s_U$ and $s_V$. The element $c_{UV}$ defines a nontrivial Čech 1-cocycle and hence a nonzero class in
$$
H^1(\mathcal{B},\mathcal{F}).
$$

This nontrivial cohomology class represents a genuine obstruction to global integration. In the interpretation of the present paper, it corresponds to a structural brain pathology: the failure is not localized in either $U$ or $V$ alone, but arises from their incompatibility.

This elementary calculation generalizes to higher-degree cohomology for more refined open covers and provides the mathematical basis for the classification of pathologies proposed in Section 7.

\section{Mathematical Mechanism of $H^1$ in Schizophrenia}

Schizophrenia is characterized by a "fragmented reality" where local perceptions fail to coalesce into a global truth. In Sheaf Theory, this is precisely the presence of a non-trivial \textbf{1-cocycle}.

Consider a collection of neural modules $\{U_i\}$. A patient may possess local sections $s_i \in \mathcal{F}(U_i)$ that are internally consistent. However, the transition functions $g_{ij} = s_i - s_j$ on the overlaps $U_i \cap U_j$ may fail to satisfy the \textbf{vanishing condition of the first cohomology group}:
\begin{equation}
    \sum_{cycles} g_{ij} \neq 0 \implies H^1(\mathcal{B}, \mathcal{F}) \neq 0
\end{equation}
This non-zero value represents a \textbf{topological hole} in the cognitive space. The "looping inconsistency" explains why paradoxical beliefs can be locally maintained but globally irreconcilable, creating the characteristic "split" in the unified self ($H^0$).

$$ $$
$$ $$

\noindent Takao Inou\'{e}

\noindent Faculty of Informatics

\noindent Yamato University

\noindent Katayama-cho 2-5-1, Suita, Osaka, 564-0082, Japan

\noindent inoue.takao@yamato-u.ac.jp
 
\noindent (Personal) takaoapple@gmail.com (I prefer my personal mail)

\bigskip

\end{document}